\documentclass{elsarticle}
\pdfoutput=1

\usepackage{url}
\usepackage{wrapfig}
\usepackage{bookman}

\usepackage{amsfonts}
\usepackage{siunitx}
\usepackage{comment}

\usepackage{amsmath,amssymb}
\usepackage{xspace}

\newcommand{\ISN}[0]{{\mathsf{ISN}}\xspace}
\newcommand{\CISN}[0]{\mathcal{C-}{\mathsf{ISN}}\xspace}

\newcommand{\set}[1]{\{#1\}}                      
                     
\usepackage{tikz}
	\usetikzlibrary{shapes,arrows,shapes.geometric,fit,calc,positioning,automata}
	\newcommand{\yslant}{0.5}
	\newcommand{\xslant}{-0.6}

\usepackage{marvosym}
\usepackage{graphicx}
\usepackage{pifont}

\usepackage{enumitem} 
\usepackage{savesym}
\savesymbol{AND}
\usepackage{algorithmic}
\usepackage{algorithm}

\usepackage[amsmath,thmmarks]{ntheorem}
\theoremseparator{.}

\theorembodyfont{\itshape}
\newtheorem{definition}{Definition}
\newtheorem{theorem}{Theorem}
\newtheorem{proposition}{Proposition}
\newtheorem{lemma}{Lemma}

\theorembodyfont{\normalfont}
\newtheorem{example}{Example}

\theoremsymbol{\ensuremath{\square}}
\theorembodyfont{\normalfont}
\theoremstyle{nonumberplain}
\theoremheaderfont{\normalfont\itshape}
\newtheorem{proof}{Proof}

\usepackage{mathtools}
\usepackage{color}
\usepackage{verbatim}

\makeatletter
\def\ps@pprintTitle{%
 \let\@oddhead\@empty
 \let\@evenhead\@empty
 \def\@oddfoot{}%
 \let\@evenfoot\@oddfoot}
\makeatother

\begin{document}

\title{Coordinating Multiagent Industrial Symbiosis}
\author{Vahid Yazdanpanah\fnref{myfootnote}}
\fntext[myfootnote]{Corresponding author. E-mail address: V.Yazdanpanah@soton.ac.uk (Vahid Yazdanpanah).}
\address{Department of Electronics and Computer Science, \\ University of Southampton, University Road, SO17 1BJ, Southampton, United Kingdom.}
\author{Devrim Murat Yazan}
\author{W. Henk M. Zijm}
\address{Department of Industrial Engineering and Business Information Systems, \\ University of Twente, Drienerlolaan 5, 7522 NB, Enschede, The Netherlands.}

\begin{abstract}
We present a formal multiagent framework for coordinating a class of collaborative industrial practices called ``Industrial Symbiotic Networks ($\ISN$s)'' as cooperative games. The  game-theoretic formulation of $\ISN$s enables systematic reasoning about what we call the $\ISN$ implementation problem. Specifically, the characteristics of $\ISN$s may lead to the inapplicability of standard  \emph{fair} and \emph{stable} benefit allocation methods. Inspired by realistic $\ISN$ scenarios and following the literature on normative multiagent systems, we consider \emph{regulations} and normative socio-economic \emph{policies} as coordination instruments that in combination with $\ISN$ games resolve the situation. In this multiagent system, employing Marginal Contribution Nets (MC-Nets) as rule-based cooperative game representations foster the combination of regulations and $\ISN$ games with no loss in expressiveness. We develop algorithmic methods for generating regulations that ensure the implementability of $\ISN$s and as a policy support, present the policy requirements that guarantee the implementability of all the desired $\ISN$s in a \emph{balanced-budget} way. 
\end{abstract}


\maketitle

\section{Introduction}
Industrial Symbiotic Networks ($\ISN$s) are mainly seen as collaborative networks of industries with the aim to reduce the use of virgin material by circulating reusable resources (e.g., physical waste material and  energy) among the network members \cite{chertow2000industrial,lombardi2012redefining,yazan2016design}. In such networks, symbiosis leads to socio-economic and  environmental benefits for involved industrial agents and the society (see \cite{d2020environmental,shen2020can}). One barrier against stable $\ISN$ implementations is the lack of  frameworks able to secure such networks against unfair and unstable allocation of obtainable benefits among the involved industrial firms. In other words, although in general $\ISN$s result in the reduction of the total cost, a remaining challenge for operationalization of $\ISN$s is to tailor reasonable mechanisms for allocating the total obtainable cost reductions---in a  fair and stable manner---among the contributing firms. Otherwise, even if economic benefits are foreseeable, lack of stability and/or fairness may lead to non-cooperative  decisions. This will be the main focus of what we call the \emph{industrial symbiosis implementation} problem. Reviewing recent contributions in the field of industrial symbiosis research, we encounter  studies focusing  on the necessity to consider  interrelations between industrial enterprises \cite{yazan2016design,yazdanpanah2019fisof} and the role of contract settings in the process of $\ISN$ implementation \cite{albino2016exploring,BSE2020}. We believe that a missed element for shifting from \emph{theoretical} $\ISN$ design to \emph{practical} $\ISN$ implementation  is to model, reason about, and support $\ISN$ decision processes in a dynamic way (and not by using snapshot-based modeling frameworks). 

For such a multiagent setting, the mature field of cooperative game theory    provides rigorous  methodologies and established solution concepts, e.g. the core of the game and the Shapley allocation \cite{driessen2013cooperative,mas1995microeconomic,osborne1994course,borm2001operations}. However, for  $\ISN$s modeled as a cooperative game, these established solution concepts may be either non-feasible (due to properties of the game, e.g. being \emph{unbalanced}) or non-applicable (due to properties that the industrial domain asks for but solution concepts cannot ensure, e.g. individual as well as collective rationality). This calls for contextualized multiagent solutions that take into account both the complexities of $\ISN$s  	 and the characteristics of the employable game-theoretical solution concepts. Accordingly, inspired by realistic $\ISN$ scenarios and following the literature on normative multi-agent systems \cite{shoham1995social,grossi2013norms,andrighetto2013normative}, we consider \emph{regulative} rules and normative socio-economic \emph{policies} as two elements that in combination with $\ISN$ games result in the introduction of the novel concept of \emph{Coordinated} $\ISN$s ($\CISN$s)\footnote{See \cite{wooldridge2009introduction} for multiagent solution concepts in general and \cite{bussmann2013multiagent,fathalikhani2020government} for their application in the industrial domain.}. We formally present regulations as monetary incentive rules to enforce desired industrial collaborations with respect to an established policy. Regarding our representational approach, we use  Marginal Contribution Nets (MC-Nets) as rule-based cooperative game representations. This simply fosters the combination of regulative rules and $\ISN$ games with no loss in expressiveness. Accordingly, applying regulatory rules to $\ISN$s enables $\ISN$ policy-makers to transform $\ISN$ games and ensure the implementability of desired ones in a fair and stable manner. 

In this work, we provide a coordinated multiagent system---using MC-net cooperative games---for the implementation phase of $\ISN$s. Moreover, we develop algorithmic methods for generating regulations that ensure the implementability of  $\ISN$s. Finally, as a policy support, we show the $\ISN$ policy requirements that guarantee  the implementability of all the desired industrial collaborations in a \emph{balanced-budget} way.  

This paper is  structured as follows. Section \ref{sec:analysis} provides a conceptual analysis on $\ISN$s and allocation problems in such multiagent collaborative networks. Section \ref{sec:prelim} introduces preliminary formal notions and game theoretic solution concepts required for our $\ISN$ implementation framework. Sections \ref{sec:ISN_games} and \ref{sec:5} present our $\ISN$ frameworks and illustrate the verified results on effectivity of  the developed coordinated multiagent system for implementing industrial symbiosis. Finally, Section \ref{sec:conclusions} concludes the paper by highlighting the main contributions  and potential extensions of this work.

\section{Conceptual Analysis} \label{sec:analysis}
In this section, we (1) present the intuition behind our approach using a running example, (2) discuss our norm-based perspective for capturing $\ISN$ regulations, (3) describe the evaluation criteria for an ideal $\ISN$ implementation framework, and 4) review previous work on tailoring game-theoretic solution concepts for industrial symbiosis implementation problem.

\subsection{ISN as a Multiagent Practice}
To explain the dynamics of implementing $\ISN$s as multiagent industrial practices, we use a running example. Imagine three industries $i$, $j$, and $k$ in an industrial park such that $r_i$, $r_j$, and $r_k$ are among recyclable resources in the three firms' wastes, respectively. Moreover, $i$, $j$, and $k$ require $r_k$, $r_i$, and $r_j$ as their primary inputs, respectively. In such scenarios, discharging wastes and purchasing traditional primary inputs are transactions that incur cost. Hence, having the chance to reuse a material,  firms prefer recycling and transporting reusable resources to other enterprises if such transactions result in obtainable cost reductions for both parties---meaning that  it reduces the related costs for discharging wastes (on the resource provider side) and purchasing cost (on the resource receiver side). On the other hand, the implementation of such an industrial network involves transportation, treatment, and transaction costs. In principle, aggregating resource treatment processes using  refineries, combining transaction costs, and coordinating joint transportation may lead to   significant cost reductions at the collective level. 

What we call the  \emph{industrial symbiosis implementation} problem focuses on challenges---and accordingly seeks solutions---for sharing this \emph{collectively obtainable} benefit among the involved firms. Simply stated, the applied method for distributing the total obtainable benefit among involved agents is crucial while reasoning about implementing an $\ISN$. Imagine a scenario in which  symbiotic relations $ij$, $ik$, and $jk$, respectively  result in $4$, $5$, and $4$ utility units of benefit, the symbiotic network $ijk$ leads to $6$ units of benefit, and each agent can be involved in at most one symbiotic relation. To implement the $ijk$ $\ISN$, one main question is about the method for distributing the  benefit value $6$ among the three agents such that they all be induced to implement this $\ISN$. For instance, as $i$ and $k$ can obtain $5$ utils together, they will defect the $\ISN$ $ijk$ if we divide  the $6$ units of util equally ($2$ utils to each agent). Note that allocating benefit values  lower than their ``traditional"  benefits---that is obtainable in case firms defect the collaboration---results in unstable $\ISN$s.Moreover, unfair mechanisms that disregard the contribution of firms may cause the firms to move to other $\ISN$s that do so. In brief, even if an $\ISN$ results in sufficient cost reductions (at the collective level), its implementation and applied allocation methods determine whether it will be realized and maintained. Our main objective in this work is to provide a multiagent implementation framework for $\ISN$s  that enables fair and stable allocation of obtainable benefits. In further sections, we review two standard allocation methods, discuss their applicability for benefit-sharing in $\ISN$s, and introduce our normatively-coordinated multiagent system to guarantee stability and fairness in $\ISN$s.  

\subsection{ISN Regulations as socio-economic Norms} 
In real cases, $\ISN$s take place under regulations that concern environmental as well as societal policies. Hence, industrial agents have to comply to a set of rules. For instance, avoiding waste discharge may be encouraged (i.e., normatively promoted) by the local authority or transporting a specific type of hazardous waste may be forbidden (i.e., normatively prohibited) in a region. Accordingly, to nudge the collective behavior, monetary incentives in the form of subsidies and taxes are well-established solutions. This shows that the $\ISN$ implementation problem is not only about decision processes among strategic utility-maximizing industry representatives (at a microeconomic level) but in addition involves regulatory dimensions---such as presence of binding/encouraging monetary incentives (at a macroeconomic level). 

\begin{figure}[!htb]
\centering   

\begin{tikzpicture}[scale=.70,every node/.style={minimum size=.6cm},on grid]

	\begin{scope}[
		yshift=-120,
		every node/.append style={yslant=\yslant,xslant=\xslant},
		yslant=\yslant,xslant=\xslant
	] 
		
		\draw[black, dashed, thin] (0,0) rectangle (7,7); 
      		\draw[dotted, thin] (5,2) -- (2,2); 
           \draw[dotted, thin] (2,5) -- (2,2); 
          \draw[dotted, thin] (5,2) -- (2,5);

		\draw[fill=red]  
			(5,2) circle (.1) %
		(2,2) circle (.1) %
            			(2,5) circle (.1); %
\node[coordinate] (A) at (5,2) {};
\node[coordinate] (B) at (2,2) {};
\node[coordinate] (C) at (2,5) {};

		\fill[black]
			(0.1,6.5) node[right, scale=.7] {Microeconomic level: $\ISN$}	
			(5.1,1.9) node [below,scale=.7]{\textbf{Agent $i$}}
			(1.9,1.9) node [below,scale=.7]{\textbf{Agent $j$}}
            (1.9,5.1) node[above,scale=.7]{\textbf{Agent $k$}}

            ;
 
    \draw [rounded corners,fill=gray!20] (2.25,3.5)--(3.25,3.5)--(3.25,2.5)--cycle ;
   \node  (AM) at (3.0,3.25) {$A$};
     \draw[-latex, thin] (5,2.1) to [out=120,in=30] (3.25,3.5);
     \draw[-latex, thin] (2.1,2) to [out=10,in=-60] (3.25,2.5);
     \draw[-latex, thin] (2,4.9) to [out=-120,in=150] (2.25,3.5);

	\end{scope}

	\draw[ultra thin](3.8,4) to (3.8,-0.32);
	\draw[ultra thin](.8,2.4) to (.8,-1.8);
    \draw[ultra thin](-1.0,4.4) to (-1.0,0.28);

	\begin{scope}[
		yshift=0,
		every node/.append style={yslant=\yslant,xslant=\xslant},
		yslant=\yslant,xslant=\xslant
	]
		
		\fill[white,fill opacity=.75] (0,0) rectangle (7,7); 
		\draw[black, dashed, thin] (0,0) rectangle (7,7); 
      		\draw[dotted, thin] (5,2) -- (2,2); 
           \draw[dotted, thin] (2,5) -- (2,2); 
          \draw[dotted, thin] (5,2) -- (2,5);

		\draw [fill=red]
			(5,2) circle (.1) %
			(2,2) circle (.1) %
			(2,5) circle (.1) %
            	(5,5) circle (.25); %
		\fill[black]
			(0.1,6.5) node[right, scale=.7] {Macroeconomic level: Coordinated $\ISN$}
			(5.1,1.9) node[below,scale=.7]{\textbf{Agent $i$}}
			(1.9,1.9) node[below,scale=.7]{\textbf{Agent $j$}}
            (1.9,5.1) node[above,scale=.7]{\textbf{Agent $k$}}
             (5.00,4.2) node[below,scale=.65,rotate=27]{\footnotesize{Coordination Mechanism}}

                        (5.0,5.2) node[above,scale=.7]{\textbf{Regulatory Agent $R$}};

      \draw [rounded corners,fill=gray!20] (2.0,3.5)--(3.5,3.5)--(3.5,2)--cycle ;       
   \node  (AM) at (3.2,3.0) {$\hat{A}$};
  \draw[-latex, thin] (5.25,5.0) to [out=-30,in=30] (3.5,3.5);

	\end{scope} 
\end{tikzpicture}
  
 \caption{At the microeconomic  level, $A$ represents the set of all benefit allocation methods that are preferable for all the firms. At the macroeconomic level, due to the introduced coordination mechanism by the regulatory agent (respecting the established socio-economic policy), we have the allocation set $\hat{A}$ either equal to $A$ or as a shrunk/extended version of it.}

\label{fig:01}
\end{figure}
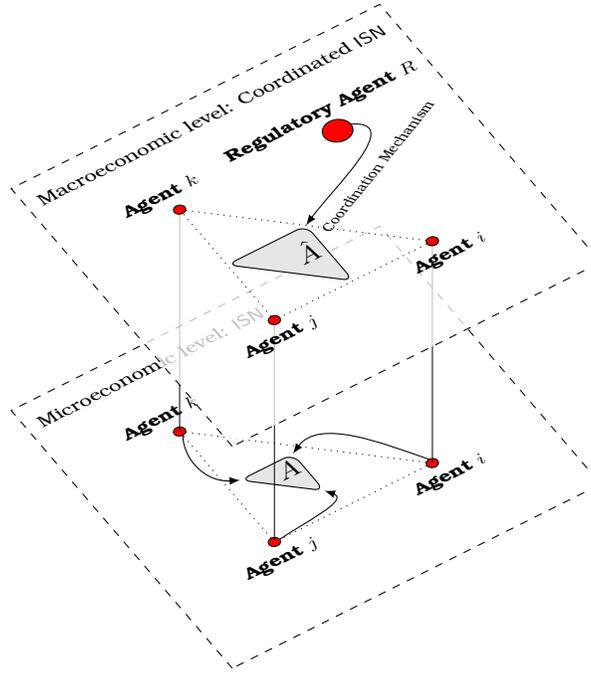

To capture the regulatory dimension of $\ISN$s, we apply a normative policy that respects the socio-economic as well as environmental desirabilities  and categorizes possible coalitions of industries in three classes of: \emph{promoted}, \emph{permitted}, and \emph{prohibited}. Accordingly, the regulatory agent respects this classification and allocates  incentives such that industrial agents will be induced to: implement promoted $\ISN$s and avoid prohibited ones (while permitted $\ISN$s are neutral from  the policy-maker's point of view).  For instance, in our $\ISN$ scenario, allocating $10$ units of incentive to $ijk$ and $0$ to other possible $\ISN$s induces all the rational agents to form the grand coalition and implement $ijk$---as they cannot benefit more in case they defect. We call the $\ISN$s that take place under regulations, \emph{Coordinated $\ISN$s} ($\CISN$s). Note that the term  ``coordination''  in this context refers to the application and efficacy of monetary incentive mechanisms in the $\ISN$ implementation phase, and should not be confused with $\ISN$ administration (i.e., managing the evolution of relations). Figure \ref{fig:01} presents a  schematic view on the role of the regulatory agents in $\CISN$s.

\subsection{Evaluation Criteria for ISN Implementation Frameworks} \label{sec:evaluation}

Dealing with firms that perform in a complex multiagent industrial context calls for implementation platforms that can be tuned to specific settings, can be scaled for implementing various $\ISN$ topologies,  do not require industries to sacrifice  financially, and allow industries to practice their freedom in the market. We deem that the quality of an $\ISN$ implementation framework should be evaluated by (1) \emph{ Generality} as the level of flexibility in the sense of independence from  agents' internal reasoning processes (i.e., how much the framework adheres to the principle of \emph{separation of concerns}), (2) \emph{Expressivity} as the level of scalability in the sense of independence from size and topology of the network, (3) \emph{Rationality} as the level that the employed allocation mechanisms comply to the collective as well as individual rationality axiom (i.e., the framework should assume that no agent (group) participates in a cooperative practice if they expect higher utility otherwise), and (4) \emph{Autonomy} as the level of allowance (i.e., non-restrictiveness) of the employed coordination mechanisms. Then an ideal framework for implementing $\ISN$s should be general---i.e., it should allow for manipulation  in the sense that the network designer  does not face any re-engineering/calibration burden---sufficiently expressive, rationally acceptable for all firms, and respect their autonomy. The goal of this paper is to develop an implementation framework for $\ISN$s that has properties close to the ideal one.

\subsection{Previous Work at a Glance}
The idea of employing cooperative game theory for analysis and implementation of  industrial symbiosis have only been sparsely explored \cite{grimes2007game,chew2009game,iesm2017}. In \cite{grimes2007game}, Grimes-Casey et al. used both cooperative and non-cooperative game theory for analyzing the behavior of firms engaged in a case-specific industrial ecology. While the analysis is expressive, the implemented relations are specific to refillable/disposable bottle life-cycles. In \cite{chew2009game}, Chew et al. tailored a mechanism for allocating costs among participating agents that expects an involved industry to ``bear the extra cost''. Although such an approach results in collective benefits, it is not in-line with the individual rationality axiom. In \cite{iesm2017}, Yazdanpanah and Yazan model bilateral industrial symbiotic relations as cooperative games and show that in such a specific class of symbiotic relations, the total operational costs can be allocated fairly and stably. Our work relaxes the limitation on the number of involved industries and---using the concept of Marginal Contribution Nets (MC-Nets)---enables a representation that is sufficiently expressive to capture the regulatory aspect of $\ISN$s. We will give a more detailed review of these papers in Section \ref{sec:revisit} after covering the technical background.

\section{Preliminaries} \label{sec:prelim}
In this section, we recall the preliminary notions in cooperative games, the MC-Net representation of such games, and the two principal solution concepts: the Shapley value and the Core\footnote{The presented material on basics in cooperative games is based on \cite{osborne1994course,mas1995microeconomic} while for the MC-Net notations, we build on \cite{ieong2005marginal,lesca2017coalition}.}. Moreover, we discuss in more detail the technical aspects of previous work that applied game-theoretical methods for $\ISN$ modeling and analysis. 

\subsection{Technical Background}
In this work, we build on the \emph{transferable utility} assumption in multiagent settings. This is to assume that the payoff to a group of agents involved in an $\ISN$ (as a cooperative practice) can be freely distributed among group members.  

\noindent \emph{Cooperative Games:} 
Multiagent cooperative games with transferable utility are often modeled by the tuple $(N,v)$, where $N$ is the finite set of agents and $v : 2^N \mapsto \mathbb{R}$ is the characteristic function that maps each possible agent group $S \subseteq N$ to a real-valued payoff $v(S)$. In such games, the so-called \emph{allocation problem} focuses on methods to distribute $v(S)$ among all the agents (in $S$) in a reasonable manner. That is, $v(S)$ is the result of a potential cooperative practice, hence ought to be distributed among agents in $S$ such that they all be induced to cooperate (or remain in the cooperation). Various solution concepts specify the utility each agent receives by taking into account properties like fairness and stability.  The two standard solution concepts that characterize fair and stable allocation of benefits are the Shapley value and the Core, respectively.  

\noindent \emph{Shapley Value:} 
The Shapley value prescribes a notion of \emph{fairness}. It says that assuming the formation of the grand coalition $N=\set{1,\dots,n}$, each agent $i \in N$ should receive its average marginal contribution  over all possible permutations of the agent groups. Let $s$ and $n$, represent the cardinality of $S$ and $N$, respectively. Then, the Shapley value of $i$ under characteristic function $v$, denoted by $\Phi_i(v)$, is formally specified as  $\Phi_i(v)=\sum_{S\subseteq N\setminus\set{i}} \frac{s! (n-s-1)!}{n!} (v(S\cup \{i\})-v(S))
$. For a game $(N,v)$, the unique list of real-valued payoffs $x=(\Phi_1(v),\cdots,\Phi_n(v)) \in \mathbb{R}^n$ is called the Shapley allocation for the game. The Shapley allocation have been extensively studied in the game theory  literature and satisfies various desired properties in  multi-agent practices. Moreover, it can be axiomatized using the following properties.

\begin{itemize}[leftmargin=*]
\item \emph{Efficiency (EFF):} The overall available utility $v(N)$ is allocated to the agents in $N$, i.e., $\sum\limits_{i \in N} \Phi_i(v)=v(N)$. 
\item \emph{Symmetry (SYM):} Any arbitrary agents $i$ and $j$ that make the same contribution receive the same payoff, i.e., $\Phi_i(v)=\Phi_j(v)$.
\item \emph{Dummy Player (DUM):} Any arbitrary agent $i$ of which  its marginal contribution to each group $S$ is the same, receives the payoff that it can earn on its own; i.e., $\Phi_i(v)=v(\{i\})$.
\item \emph{Additivity (ADD):} For any two cooperative games $(N,v)$ and $(N,w)$, $\Phi_i(u+w)=\Phi_i(v)+\Phi_i(w)$ for all $i \in N$, where for all $S \subseteq N$, the characteristic function  $v+w$ is defined as $(v+w)(S)=v(S)+w(S)$.  
\end{itemize}

In the following, we refer to an allocation that satisfies these properties as a \emph{fair} allocation. 

\noindent \emph{Core of the Game:} 
In core allocations, the focus is on the notion of \emph{stability}. In brief, an allocation is stable if no agent (group) benefits by defecting the cooperation. Formally, for a game $(N,v)$, any list of real-valued payoffs $x \in \mathbb{R}^n$ that satisfies the following conditions is a core allocation for the game:

\begin{itemize}[leftmargin=*]
\item \emph{Rationality (RAT):} $\forall S \subseteq N : \sum\limits_{i \in S} x_i \geq v(S)$ 
\item \emph{Efficiency (EFF):} $\sum\limits_{i \in N}x_i=v(N)$
\end{itemize}

One main question is whether for a given game, the core is non-empty (i.e., that there exists a stable allocation for the game). A game for which there exist a non-empty set of stable allocations should satisfy the \emph{balancedness} property, defined as follows. Let $1_S \in \mathbb{R}^n$ be the membership vector of $S$, where $(1_S)_i=1$ if $i \in S$ and $(1_S)_i=0$ otherwise. Moreover, let $(\lambda_S)_{S\subseteq N}$ be a vector of weights $\lambda_S \in [0,1]$. A vector  $(\lambda_S)_{S\subseteq N}$ is a \emph{balanced} vector if for all $i \in N$, we have that $\sum_{S \subseteq N} \lambda_S (1_S)_i=1$. Finally, a game is \emph{balanced} if for all balanced vectors of weights, we have that $\sum_{S \subseteq N} \lambda_S v(S) \leq v(N)$. According to the Bondereva-Shapley theorem, a game has a non-empty core if and only if it is balanced \cite{shapley1967balanced,bondareva1963some}.  

In the following, we refer to an allocation that satisfies RAT and EFF as a \emph{stable} allocation. 

\noindent \emph{Marginal Contribution Nets (MC-Nets):} 
Representing cooperative games by their characteristic functions (i.e., specifying values $v(S)$ for all the possible coalitions $S \subseteq N$) may become unfeasible in large-scale applications. In this work, as we are aiming to implement $\ISN$s in a scalable manner, we employ a  \emph{basic MC-Net} \cite{ieong2005marginal} representation that uses a set of rules to specify the value of possible agent coalitions. Moreover, attempting to capture the regulatory aspect of $\ISN$s makes employing rule-based game representations a natural approach. 

A basic MC-Net represents the cooperative game among agents in $N$ as a finite set of rules $\{\rho_i : (\mathcal{P}_i,\mathcal{N}_i) \mapsto v_i  \}_{i \in K}$, where $\mathcal{P}_i  \subseteq N$, $\mathcal{N}_i \subset N$, $\mathcal{P}_i \cap \mathcal{N}_i = \emptyset$, $v_i \in \mathbb{R} \setminus \{0\}$, and $K$ is the set of rule indices. For an agent coalition $S\subseteq N$, a rule $\rho_i$ is \emph{applicable} if $\mathcal{P}_i \subseteq S $ and $\mathcal{N}_i \cap S=\emptyset $ (i.e., $S$ contains all the agents in $\mathcal{P}_i$ and no agent in $\mathcal{N}_i$). Let $\Pi(S)$ denote the set of rule indices that are applicable to $S$. Then the value of $S$, denoted by $v(S)$, will be equal to $\sum_{i \in \Pi(S)} v_i$.  In further sections, we present an MC-Net representation of the $ijk$ $\ISN$ scenario and illustrate how this rule-based representation enables applying norm-based coordination to $\ISN$s. 

\subsection{Revisiting Previous Work} \label{sec:revisit}
Chew et al. in \cite{chew2009game} analyze the interaction of participating companies in an Eco-industrial park seeking to develop a game-theoretic implementation framework for inter-plant water integration. In their cooperative game model, by assuming the compliance of agents to their commitments, the optimum collective benefit is achievable. As the authors mention, in case the cooperation takes place, their allocation mechanism results in  higher collective payoff in comparison to their non-cooperative game scheme. This result is achieved through adding contextualized interaction protocols that compel  the industries to  act in a desired manner. Roughly speaking, it is  assumed that the network manager has control over internal operations and decision processes of involved agents (which may be applicable in specific case studies but is in contrast with the principle  of \emph{separation of concerns}).  For instance, given the availability of an optimal \emph{wastewater interchange  scheme}, it is shown that in case the agents adopt the scheme and act accordingly, they can benefit both individually and collectively. In other words, the focus is shifted towards providing methods for optimizing the scheme in a specific case.

In a more recent work, Yazdanpanah and Yazan looked into the modeling and implementation of industrial symbiotic relations as two person cooperative games \cite{iesm2017}. Their focus is on allocation of the total operational cost among involved agents using a tailored version of the Shapley value and the standard notion of core. They show that for industrial symbiotic relation games, core is non-empty and hence such symbiotic practices are implementable in a stable manner. Moreover, as the Shapley value will be in the core, it is rational for industries to implement the Shapley allocation (with no need for interruption from the regulatory agent). Notice that although their industrial symbiosis implementation satisfies desired properties, e.g.,  autonomy and rationality, it is not expressive for implementing symbiotic relations among three or more industries. This is basically because their analysis is based on properties of two-person games.

Finally, Grimes-Casey et al. \cite{grimes2007game} focus on cooperative decision-making and heterogeneity of the involved agents (with respect to their epistemic states) in an industrial symbiosis scenario. They employ cost-based mechanisms to nudge the behavior of manufacturer as well as consumer agents towards using refillable beverage containers. Although their cooperative management framework is problem-specific, it is expressive and scalable as they employ profit values that are computable in low complexity. They also discuss that in real cases, the applicability of most cooperative game solution concepts depends  on government  enforcement. This is in-line with our attempt to capture the regulatory aspect of industrial symbiosis using incentive mechanisms.

\section{ISN Games}  \label{sec:ISN_games}
As discussed in \cite{albino2016exploring,iesm2017}, the total obtainable cost reduction---as the economic benefit---and its allocation among involved firms are key drivers behind the stability of $\ISN$s. For any set of industrial agents $S$, this total value can be computed based on the total \emph{traditional} cost, denoted by $T(S)$, and the total $\ISN$ \emph{operational} cost, denoted by $O(S)$. In brief, $T(S)$ is the summation of all the costs that firms have to pay in case the $\ISN$ does not occur (i.e., to discharge wastes and to purchase traditional primary inputs). On the other hand, $O(S)$ is the summation of costs that firms have to pay collectively in case the $\ISN$ is realized  (i.e.,  the costs for recycling and  treatment, for transporting resources among firms, and finally the transaction costs). Accordingly, for a non-empty finite set of industrial agents $S$ the obtainable symbiotic value $v(S)$ is equal to $T(S)-O(S)$. In this work, we assume a potential  $\ISN$, with a positive total obtainable value, and aim for tailoring game-theoretic value allocation and accordingly coordination mechanisms that guarantee a fair and stable implementation of the symbiosis. 

\subsection{ISNs as Cooperative Games}
Our $ijk$ $\ISN$ scenario can be modeled as a cooperative game in which $v(S)$ for any empty/singleton $S$ is $0$ and agent groups  $ij$, $ik$, $jk$, and $ijk$ have the values $4$, $5$, $4$, and $6$, respectively. Note that as the focus of $\ISN$s are on the benefit values obtainable due to potential cost reductions, all the empty and singleton agent groups have a zero value because cost reduction is meaningless in such cases. In the game theory language, the payoffs in $\ISN$ games are normalized. Moreover, the game is superadditive in nature.\footnote{Superadditivity implies that forming a symbiotic coalition of industrial agents either results in no value  or in a positive value. Implicitly, growth of a group can never result in decrease of the value.} So, given the traditional and operational cost values for all the possible agent groups $S$ (i.e., $T(S)$ and $O(S)$) in the non-empty finite set of industrial agents $N$, the $\ISN$ among agents in $N$ can be formally modeled as follows.

\begin{definition}[$\ISN$ Games] \label{def:ISR-Games} Let $N$ be a non-empty finite set of industrial agents. Moreover, for any agent group $S \subseteq N$, let $T(S)$ and $O(S)$ respectively denote the total traditional  and  operational costs for $S$. We say the $\ISN$ among industrial agents in $N$ is a normalized superadditive  cooperative game $(N,v)$ where $v(S)$ is:
\[
v(S)= 
\begin{cases}
    0,              & \text{if } |S| \leq 1 \\    
    T(S)- O(S),      & \text{otherwise}
\end{cases}
\]
\end{definition}

According to the following proposition, basic MC-Nets can be used to represent $\ISN$s. In further sections, this  representation aids combining $\ISN$ games with normative coordination rules. 

\begin{proposition}[$\ISN$s as MC-Nets] \label{prop:ISN-MCNET}
Any $\ISN$ can be represented as a basic MC-Net. 
\end{proposition}
\begin{proof} 
We provide a constructive proof by (1) introducing an algorithm for specification of all the required MC-Net rules and (2) showing that the constructed MC-Net is equal to the original $\ISN$ game. {\it(1) -} Let $(N,v)$ be an arbitrary $\ISN$ game among industrial agents in $N$. Moreover, let $S_{\geq 2}= \{S \subseteq N : |S|\geq 2 \}$ be the set of all agent groups with two or more members and let $K=|S_{\geq 2}|$ denote its cardinality. We start with an empty set of rules. Then for all agent groups $S_{i} \in S_{\geq 2}$, for $i= 1,\dots,K$, we add a rule $\{\rho_i : (S_i,N \setminus S_i) \mapsto v_i=T(S_i)-O(S_i)  \}$. {\it(2) -} As in all the constructed rules $\rho_i$  it holds that $\mathcal{P}_i \cap \mathcal{N}_i=\emptyset$ and $\mathcal{P}_i \cup \mathcal{N}_i = N$, we have that $\sum_{i\in \Pi(S)} v_i$ is equal to $v(S)$ for all the members of $S_{\geq 2}$. Moreover, $\Pi(S)$ for empty and singleton agent groups would be empty, hence reflects the $0$ value for such groups in the original game.     \end{proof}

Note that the proof does not simply rely on the representation power and expressivity of MC-Nets (as shown in \cite{ieong2005marginal}) but provides a constructive method that respects the context of industrial symbiosis and related cost values to generate all the required rules for representing $\ISN$s as MC-Nets.  
\begin{example}[$\ISN$ Scenario] \label{ex:mcnet01} 
Our running example can be   represented by the  basic MC-Net\footnote{For notational simplicity,  we avoid brackets around agent groups, e.g., we write $ij$ instead of $\{i,j\}$.}  $\{\rho_1 : (ij, k) \mapsto 4,  \rho_2 : (ik, j) \mapsto 5, \rho_3 : (jk, i) \mapsto 4, \rho_4 : (ijk, \emptyset) \mapsto 6 \}$. 
\end{example}

\subsection{Benefit Allocation Mechanisms and ISN Games}  
As discussed earlier, how firms share the obtainable $\ISN$ benefits plays a key role in the process of $\ISN$ implementation, mainly due to stability and fairness concerns. Roughly speaking, industrial firms are economically rational firms that defect non-beneficial relations (instability) and mostly tend to reject $\ISN$ proposals in which benefits are not shared with respect to their contribution (unfairness). In this work, we focus on Core- and Shapley-allocation mechanisms as two standard methods that characterize stability and fairness in cooperative games, receptively. We show that these solution concepts are applicable in a specific class of $\ISN$s but are not generally scalable for value allocation in the implementation phase of $\ISN$s. This motivates introducing incentive mechanisms to guarantee the implementability of ``desired'' $\ISN$s.  

\subsubsection{Two-Person Industrial Symbiosis Games}
When the game is between two industrial firms (i.e., a bilateral relation between a resource receiver/provider couple), it has additional properties that result in applicability of both Core and Shapley allocations. We denote the class of such $\ISN$ games by $\ISN_{\Lambda}$. This is, $\ISN_{\Lambda}= \set{ (N,v) : (N,v)  \text{ is an $\ISN$ game and }   |N|=2}$. Moreover, the $\ISN$ games in which three or more agents are involved will form $\ISN_\Delta$. The class of $\ISN_\Lambda$ games corresponds to the so called ISR games in \cite{iesm2017}. The difference is on the value allocation perspective as in \cite{iesm2017}, they assume the elimination of traditional costs (thanks to implementation of the symbiotic relation) and focus  on the allocation of operational costs; while we focus on the allocation of the total benefit, obtainable due to potential cost reductions.

\begin{lemma}[$\ISN_\Lambda$ Balancedness] \label{lemma01}
Let $(N,v)$ be an arbitrary $\ISN_\Lambda$ game. It always holds that $(N,v)$ is balanced. 
\end{lemma}
\begin{proof}
We show that any $\ISN_\Lambda$ game is supermodular which directly implies balancedness. A game $(N,v)$ is supermodular iff for any couple of arbitrary agent groups $S,T \subseteq N$, we have $v(S)+v(T) \leq v(S\cup T) + v(S\cap T)$. In $\ISN_\Lambda$ games, by checking the validity of this inequality for all the six possible $S,T$ combinations, the claim will be proved. For $S=\emptyset$, we have the following valid inequality $v(\emptyset)+v(T) \leq v(\emptyset \cup T = T) + v(\emptyset)$. For $S=N$, the inequality can be reformulated in the following valid form $v(N)+v(T) \leq v(N\cup T=N) + v(N\cap T=T)$. Finally, when $S$  and $T$ are equal to the only  possible  (disjoint) singleton  groups, we have $v(S)+v(T) \leq v(N) + v(\emptyset)$ which holds thanks to the superadditivity of $\ISN$ games.    
\end{proof}

Relying on Lemma \ref{lemma01}, we have the following result that focuses on the class of  $\ISN_\Lambda$ relations and shows the applicability of two standard game-theoretic solution concepts for implementing fair and stable industrial symbiotic networks. 

\begin{theorem} [Fair and Stable $\ISN_\Lambda$ Games]
Let $(N,v)$ be an arbitrary $\ISN_\Lambda$ game. The symbiotic relation among industrial agents in $N$ is implementable in a unique stable and fair manner.
\end{theorem}
\begin{proof} 
\emph{Stability}: As discussed earlier, core allocations guarantee the stability conditions (i.e., RAT and EFF). However, the core is only an applicable solution concept for \emph{balanced} games. According to Lemma \ref{lemma01}, we have that $\ISN_\Lambda$ games are balanced. Hence, the core of any arbitrary $\ISN_\Lambda$ game is nonempty and any allocation in the core guarantees the stability. \emph{Stability and Fairness}: As presented earlier, the Shapley allocation guarantees the fairness conditions (i.e., EFF, SYM, DUM, ADD). However, it does not always satisfy the rationality (RAT) condition (which is necessary for stability). According to Lemma \ref{lemma01}, we have that $\ISN_\Lambda$ games are balanced. Moreover, according to \cite[Theorem 7]{shapley1971cores}, in balanced games, the Shapley allocation is a member of the core and hence satisfies the rationality condition. Accordingly, for any $\ISN_\Lambda$ game, the Shapley allocation guarantees both the stability and fairness.     
\end{proof}

\subsubsection{ISN Games}
In this section we focus on $\ISN_\Delta$ games as the class of $\ISN$ games with three or more participants and discuss the applicability of the two above mentioned allocation mechanisms for implementing such industrial games. 

\begin{example} [Neither Core Nor Shapley] \label{ex:02} 
Recall the $ijk$ $\ISN_\Delta$  scenario from Section \ref{sec:analysis}. To have a stable allocation $(x_i,x_j,x_k)$ in the core, the EFF condition implies $x_i+x_j+x_k=6$ while the RAT condition implies $x_i+x_j\geq 4     \wedge  x_i+x_k \geq 5        \wedge  x_j+x_k \geq 4$. As these conditions cannot be satisfied  simultaneously, we can conclude that the core is empty and there exists no way to implement this $\ISN$ in a stable manner. Moreover, although the Shapley allocation provides a fair allocation $(13/6, 10/6, 13/6)$, it is not rational for firms to implement the $\ISN$. E.g., $i$ and $k$ obtain $30/6$ in case they defect while according to the Shapley allocation, they ought to sacrifice as they collectively  have $26/6$.
\end{example}

As illustrated in this example, the Core of $\ISN_\Delta$ games may be empty which implies the inapplicability of this solution concept as a general method for implementing $\ISN$s. We now generalize the exemplified idea to the following nonexistence theorem about implementability of  $\ISN_\Delta$ games in a fair and stable manner. 

\begin{theorem}[Unimplementability of $\ISN_\Delta$ Games] \label{th:unimp}
Let $(N,v)$ be an arbitrary $\ISN_\Delta$ game. The symbiotic relation among industrial agents in $N$ is not generally implementable in a stable manner. 
\end{theorem}
\begin{proof}
Although all $\ISN_\Delta$ games are superadditive and hence result in a positive obtainable benefit, they may be unbalanced (as illustrated in the running example). Accordingly, for any unbalanced  $\ISN_\Delta$ game, the Core is empty. In such cases, the symbiotic relation is not implementable in a stable manner.    \end{proof}

Note  that the fair implementation of $\ISN_\Delta$ games is not always in compliance with the rationality condition.  This theorem---in accordance  with the intuition presented in example \ref{ex:02}---shows that we lack general methods that guarantee stability and fairness of $\ISN$ implementations. So, even if an industrial symbiotic practice could result in \emph{collective} economic and environmental benefits, it may not last due to instable or unfair implementations. One natural response which is in-line with realistic  $\ISN$ practices is to employ monetary incentives as a means of coordination. 

\section{Coordinated ISN} \label{sec:5}
In realistic $\ISN$s, the symbiotic practice takes place in the presence of economic, social, and environmental \emph{policies} and under \emph{regulations} that aim to enforce the policies by nudging the behavior of agents towards desired ones. In other words, while the policies generally indicate whether an $\ISN$ is ``good (bad, or neutral)", the regulations are a set of norms that---in case of agents' compliance---result in a spectrum of  acceptable (collective) behaviors. Note that the acceptability, i.e., goodness, is evaluated and ought to be verified from the point of view of the policy-makers as community representatives. In this section, we follow this normative approach and aim for using normative coordination to guarantee the implementability of desirable $\ISN$s in a stable and fair manner\footnote{In the following, we simply say \emph{implementability} of $\ISN$s instead of \emph{implementability in a fair and stable} manner.}.

\subsection{Normative Coordination of ISNs}
Following \cite{shoham1995social,grossi2013norms}, we see that during the process of $\ISN$ implementation as a game, norms can be employed as game transformations, i.e., as ``ways of transforming existing games in order to bring about outcomes that are more desirable from a welfaristic point of view"\cite{grossi2013norms}. For this account,  given the economic, environmental, and social dimensions and with respect to potential socio-economic consequences, industrial symbiotic networks can be partitioned in three classes, namely \emph{promoted}, \emph{permitted}, and \emph{prohibited} $\ISN$s. Such a classification  can be modeled by a normative socio-economic policy function $\wp : 2^N \mapsto \set{p^+,p^\circ , p^-}$, where $N$ is the finite set of industrial firms. Moreover,  $p^+$, $p^\circ$, and  $p^-$ are labels---assigned by a third-party authority---indicating that the $\ISN$ among any given agent group is either promoted, permitted, or prohibited, respectively. The three sets $P_\wp^+$, $P_\wp^\circ$, and $P_\wp^-$ consist of all the $\wp$-promoted, -permitted, and -prohibited agent groups, respectively. Formally $P_\wp^+=\set{S \subseteq N : \wp(S)=p^+}$ ($P_\wp^\circ$ and $P_\wp^-$ can be formulated analogously).  Note that $\wp$ is independent of the $\ISN$ game among agents in $S$, its economic figures, and corresponding cost values---in general, it is independent of the  value function of the game. E.g., a symbiotic relation  may be labeled with $p^-$ by policy $\wp$---as it is focused on exchanging a  hazardous waste---even if it results in a high level of obtainable benefit. 

\begin{example} [Normative $\ISN$s] \label{ex-policy} In our $ijk$ $\ISN$ scenario, imagine a policy $\wp_1$ that assigns $p^-$ to all the singleton and two-member groups (e.g., because they discharge hazardous wastes in case they operate in one- or two-member groups) and  $p^+$ to the grand coalition (e.g., because in that case they have zero waste discharge). So, according to $\wp_1$, the $\ISN$ among all the three agents is ``desirable" while other possible coalitions lead to ``undesirable'' $\ISN$s.  
\end{example}

As illustrated in Example \ref{ex-policy}, any socio-economic policy function merely indicates the desirability of a potential  $\ISN$ among a given group of agents and is silent with respect to methods for \emph{enforcing} the implementability of promoted or unimplementability of  prohibited $\ISN$s.  Note that $\ISN_\Lambda$ games are always implementable. So, $\ISN$s'  implementability refers to the general class of $\ISN$ games including $\ISN_\Delta$ games.

The rationale behind introducing socio-economic policies for $\ISN$s is mainly to make sure that the set of promoted $\ISN$s are implementable in a fair and stable manner while prohibited ones are instable. To ensure this, in real $\ISN$ practices, the regulatory agent (i.e., the regional or national government) introduces regulations---to support the policy---in the form of monetary incentives\footnote{See \cite{DBLP:conf/ijcai/MeirRM11,DBLP:conf/atal/ZickPJ13} for similar approaches on incentivizing cooperative agent systems.}. This is to ascribe subsidies to promoted and taxes to prohibited collaborations (see \cite{kakhbod2013resource} for an implementation theory approach on mechanisms that employ monetary incentives to achieve desirable resource allocations). We follow this practice and employ a set of rules to ensure/avoid the implementability of desired/undesired $\ISN$s among industrial agents in $N$ via allocating incentives. Such a set of incentive rules can be represented by an MC-Net $\Re = \{\rho_i : (\mathcal{P}_i,\mathcal{N}_i) \mapsto \iota_i  \}_{i \in K}$ in which $K$ is the set of rule  indices. Let $\Im(S)$ denote the set of rule indices that are applicable to $S\subseteq N$. Then, the incentive value for $S$, denoted by $\iota(S)$, is defined as $\sum_{i \in \Im(S)} \iota_i$.  This is, a set of incentive rules can be represented also as a cooperative game $\Re=(N,\iota)$ among agents in $N$.  The following proposition shows that for any $\ISN$ game there exists a set of incentive rules to guarantee the implementability of the $\ISN$ in question.

\begin{proposition} [Implementability Ensuring Rules]  \label{prop:imp_ensurer}
Let $G$ be an arbitrary $\ISN$ game among industrial agents in $N$. There exists a set of incentive rules to guarantee the implementability of $G$. 
\end{proposition}
\begin{proof} 
Recall that according to Proposition \ref{prop:ISN-MCNET}, $G$ can be represented as an MC-Net. To prove the claim, we provide Algorithm \ref{algor} that takes the MC-Net representation of $G$ as the input and generates a set of rules that guarantee the implementability of $G$.

\begin{algorithm}[ht]
\begin{algorithmic}[1]
\STATE \textbf{Data:} $\ISN$ game $G= \{\rho_i : (\mathcal{P}_i,\mathcal{N}_i) \mapsto v_i  \}_{i \in K}$ among agents in $N$; $K$ the set of rule indices for $G$.
\STATE \textbf{Result:} Incentive rule set $\Re$ for $G$.
\STATE $n \gets length(K) $ and $\Re=\set{}$

\FOR{$i\gets1$ \TO $n$ }
     \IF{$i \in \Pi(N)$}
        \STATE {$\Re \gets \Re \cup \set{\rho_i : (\mathcal{P}_i,\mathcal{N}_i) \mapsto 0}$}
    \ELSE
        \STATE{$\Re \gets  \Re \cup \set{\rho_i : (\mathcal{P}_i,\mathcal{N}_i) \mapsto -v_i}$} 
    \ENDIF
\ENDFOR
\end{algorithmic}
\caption{Generating rule set $\Re$ for $\ISN$ game $G$.}
\label{algor}
\end{algorithm}

By allocating $-v_i$ to rules that are not applicable to $N$, any coalition other than the grand coalition will be faced with a tax value. As the original game is superadditive, the agents will have a rational incentive to cooperate in $N$ and the $\ISN$ is implementable in a stable manner thanks to the provided incentive rules.    
\end{proof} 

Till now, we introduced socio-economic policies and regulations as required (but not yet integrated) elements for modeling coordinated $\ISN$s. In the following section, we combine the idea behind incentive regulations and normative socio-economic policies to introduce the concept of \emph{Coordinated $\ISN$s} ($\CISN$s) as a multiagent system for implementing industrial symbiosis.  

\subsection{Coordinated ISNs}
As discussed above, $\ISN$ games can be combined with a set of regulatory rules that allocate incentives to agent groups (in the form of subsidies and taxes). We call this class of games, $\ISN$s in presence of coordination mechanisms, or \emph{Coordinated $\ISN$s} ($\CISN$s) in brief.

\begin{definition} [Coordinated $\ISN$ Games ($\CISN$)] Let $G$ be an $\ISN$  and $\Re$ be a  set of regulatory incentive rules, both as MC-Nets among industrial agents in $N$. Moreover, for each agent group $S\subseteq N$, let $v(S)$ and $\iota(S)$ denote the value of $S$ in $G$ and the incentive value of $S$ in $\Re$, respectively. We say the Coordinated $\ISN$ Game ($\CISN$) among industrial agents in $N$ is a cooperative game $(N,c)$ where for each agent group $S$, we have that $c(S)=v(S) + \iota(S)$. 
\end{definition}

Note that as both the $\ISN$ game $G$ and the set of regulatory incentive rules $\Re$ are MC-Nets among industrial agents in $N$, then for each agent group $S \subseteq N$ we have that $c(S)$ is equal to the summation of all the applicable rules to $S$ in both $G$ and $\Re$. Formally, $c(S)=\sum_{i \in \Pi(S)} v_i + \sum_{j \in \Im(S)} \iota_j$ where $\Pi(S)$ and $\Im(S)$ denote the set of applicable rules to $S$ in $G$ and $\Re$, respectively. Moreover, $v_i$ and $\iota_j$ denote the value of applicable rules $i$ and $j$ in $\Pi(S)$ and $\Im(S)$, respectively. We sometime use $G+\Re$ to denote the game $C$ as the result of incentivizing $G$ with $\Re$. The next proposition shows the role of regulatory rules in  the enforcement of socio-economic policies.

\begin{proposition} [Policy Enforcing Rules]
For any promoted $\ISN$ game $G$ under policy $\wp$, there exist an implementable  $\CISN$ game $C$.  
\end{proposition}
\begin{proof}  To prove, for any arbitrary promoted $G$, we require a set of regulatory incentive rules $\Re$ such that its combination with $G$ results in a stable $C$ implementation. The algorithm for generating such a $\Re$ is presented in the proof of Proposition \ref{prop:imp_ensurer}.    
\end{proof}

Analogously, similar properties hold for \emph{avoiding} prohibited $\ISN$s  or \emph{allowing} permitted ones. Avoiding  prohibited $\ISN$s can be achieved by making the $\CISN$ (that results from introducing regulatory incentives) unimplementable. On the other hand, allowing permitted  $\ISN$s would be simply the result of adding an empty set of regulatory rules. The presented approach for incentivizing $\ISN$s, is advisable when the policy-maker is aiming to ensure the implementability of a promoted $\ISN$ in an ad-hoc way. In other words, an $\Re$ that ensures the  implementability of a promoted $\ISN$ $G_1$ may ruin the implementability of another promoted $\ISN$ $G_2$. This highlights the importance of  some structural properties for socio-economic policies that aim to foster the implementability of desired $\ISN$s. As we discussed in Section  \ref{sec:analysis}, we aim for implementing $\ISN$s such that the rationality axiom will be respected. In the following, we focus on the subtleties of socio-economic policies that are enforced by regulatory rules. The question is, what are the requirements of a policy that can ensure the rationality of staying in desired $\ISN$s? We first show that to respect the rationality axiom, promoted agent groups should be disjoint. We illustrate that in case the policy-maker takes this condition into account, industrial agents have no economic incentive to defect an implementable promoted $\ISN$.

\begin{proposition} [Mutual Exclusivity of Promoted $\ISN$s]
Let $G_1$ and $G_2$ be arbitrary $\ISN$s, respectively among  promoted (nonempty) agent groups $S_1$ and $S_2$ under policy $\wp$ (i.e., $S_1, S_2 \in P_\wp^+$). Moreover, let $\Re_1$ and $\Re_2$ be  rule sets that ensure the implementability of $G_1$ and $G_2$, respectively. For $i \in \set{1,2}$, defecting from $\CISN$ $C_i=G_i+\Re_i$  is not economically rational for any agent $a \in S_i$ iff $S_1 \cap S_2 = \emptyset$.      
\end{proposition}
\begin{proof}
$``\Rightarrow"$: Suppose  $S_1 \cap S_2 \neq \emptyset$. Accordingly, we have an agent $a$ which is both a member of $S_1$ and $S_2$. For $a$ it is rational to defect either $S_1$ or $S_2$ as both the two $\CISN$s that are based on the two groups are implementable.    

$``\Leftarrow"$: Suppose $S_1$ and $S_2$ are disjoint promoted agent groups under $\wp$. As $\Re_1$ and $\Re_2$ can respectively ensure the implementability of these two groups and based on Proposition \ref{prop:ISN-MCNET}, we have that $\ISN$s among firms in $S_1$ and $S_2$ are both implementable in a stable manner. Hence, they satisfy the rationality axiom. Moreover, as the two agent groups share no agent, there will be no economic incentive to deviate between the two stable $\ISN$s.    
\end{proof} 

Accordingly, given a set of industrial agents in $N$ and a socio-economic policy $\wp$  we directly have that: 

\begin{proposition} [Minimality of Promoted $\ISN$s] \label{cor_exclusiv} For $n=|P_\wp^+|$  if $\bigcap\limits^{n}_{i=1} S_i \in P_\wp^+ =\emptyset$ then any arbitrary $S_i \in P_\wp^+$ is minimal (i.e., $S'_i \not\in P_\wp^+$ for any $S'_i \subset S_i$).
\end{proposition}

Roughly speaking, the exclusivity condition for promoted agent groups entails that any agent is in at most  one promoted group. Hence, deviation of agents does not lead to a larger promoted group as no promoted group is part of a promoted super-group, or contains a promoted sub-group. In the following, we show that the mutual exclusivity condition is sufficient for ensuring the implementability of all the $\ISN$s that take place among promoted groups of firms.  

\begin{theorem} [Conditioned Implementability] \label{theorem_CISN_impl}
Let $G$ be an arbitrary $\ISN_\Delta$ game under policy $\wp$ among industrial agents in $N$ and $n$ be the cardinality of $P_\wp^+$. If $\bigcap\limits^n_{i=1} S_i \in P_\wp^+=\emptyset$, then there exists a set of regulatory rules $\Re$, such that all the promoted symbiotic networks are implementable in the coordinated $\ISN$ defined by $C=G+\Re$. Moreover, any $\ISN$  among prohibited agent groups in $P_\wp^-$ will be unimplementable.  
\end{theorem}
\begin{proof} To prove, we provide a method to generate such an implementability ensuring set of rules. We start with an empty $\Re$. Then for all $n$ promoted $S_i \in P_\wp^+$, we call the provided algorithm in Proposition \ref{prop:imp_ensurer}. Each single run of this algorithm results in a $\Re_i$ that guarantees the implementability of the industrial symbiosis  among the set of firms in the promoted group  $S_i$. As the set of promoted agent groups comply to the mutual exclusivity condition, the unification of all the regulatory rules results in a general $\Re$. Formally, $\Re= \bigcup\limits_{i=1}^n \Re_i$. Moreover, as the algorithm applies taxation on non-promoted groups, no $\ISN$ among prohibited agent groups will be implementable.     
\end{proof}

\begin{example} [$ijk$ as a Normatively Coordinated  $\CISN$] \label{ex:just-tax} Recalling the  $\ISN$ scenario in Example \ref{ex-policy}, the only promoted group is the grand coalition while other possible agent groups are prohibited. To ensure the implementability of the unique promoted group and to avoid the implementability of other groups, the result of executing our algorithm is $\Re= \{\rho_1 : (ij, k) \mapsto -4,  \rho_2 : (ik, j) \mapsto -5, \rho_3 : (jk, i) \mapsto -4 \}$. In the $\CISN$ that results from adding $\Re$ to the original $\ISN$, industrial symbiosis among firms in the promoted group is implementable while all the prohibited groups cannot implement a stable symbiosis.   
\end{example}

\subsection{Realized ISNs and Budget-Balancedness}
As we mentioned in the beginning of Section \ref{sec:5}, regulations are  norms that in case of agents' compliance bring about the desired behavior. For instance, in Example \ref{ex:just-tax}, although according to the provided tax-based rules, defecting the grand coalition is not economically rational, it is probable that agents act irrationally---e.g., due to trust-/reputation-related issues---and go out of the promoted group. This results in possible normative behavior of a $\CISN$ with respect to  an established policy $\wp$.  So, assuming that based on evidences the set of implemented $\ISN$s are realizable, we have the following abstract definition of $\CISN$'s normative behavior under a socio-economic policy. 

\begin{definition} [$\CISN$'s Normative Behavior] 
Let $C$ be a $\CISN$ among industrial agents in $N$ under policy $\wp$ and let $E$ be the evidence set that includes all the implemented $\ISN$s among agents in $N$. We say  the behavior of $C$ complies to $\wp$ according to $E$ iff $E=P_\wp^+$; and violates it otherwise. 
\end{definition}

Given an $\ISN$ under a policy, we introduced a set of regulatory rules to ensure that all the promoted $\ISN$s will be implementable. However, although providing incentives makes them implementable, the autonomy of industrial agents may result in situations that not all the promoted agent groups implement their $\ISN$. So, although we can ensure the implementability of all the promoted $\ISN$s, the real behavior may deviate from a desired one. As our introduced method for guaranteeing  the implementability of $\ISN$s among promoted agent groups is mainly tax-based, if a $\CISN$ violates the policy, we end up with collectible tax values.  In such cases, our tax-based method can become a \emph{balanced-budget} monetary incentive mechanism (as discussed in \cite{guo2008optimal,li2003mechanisms,phelps2010evolutionary}) by employing a form of ``Robin-Hood'' principle and  redistributing the collected amount among promoted agent groups that implemented their $\ISN$. In the following, we provide an algorithm that guarantees budget-balancedness by means of a Shapley-based redistribution of the collectible tax value among agents that implemented promoted $\ISN$s.  

\begin{algorithm}[ht]
\begin{algorithmic}[1]
\STATE \textbf{Data:} $C=G+\Re$ the $\CISN$ game among industrial agents in $N$ under policy $\wp$ such that all the $\ISN$s among promoted groups in $P_\wp^+$ are implementable; $E$ the set of implemented $\ISN$s; The collectible tax value $\tau$.
\STATE \textbf{Result:} $\Omega_i(C,\wp)$ the distributable incentive value to $i \in N$.
\STATE  $S^+ \gets E \cap P_\wp^+$ , $S^+_{u} \gets  \bigcup\limits_{S \in S^+} S$

\FORALL{$i \in (S^+_{u},v)$ the sub-game of $G$}
    \STATE{$k \gets \Phi_i(v)$ the Shapley value of $i$ in $(S^+_{u},v)$}
    \STATE{$\Omega_i(C,\wp)= \frac{1}{v(S^+_{u})}.\tau.k$}
\ENDFOR

\end{algorithmic}
\caption{Tax Redistribution for $\CISN$ game $C$.}
\label{algor2}
\end{algorithm}

The correctness of Algorithm \ref{algor2} is established in Proposition \ref{prop:bb}.

\begin{proposition}[Budget Balancedness and Fairness]  \label{prop:bb}
Let $C=G+\Re$ be a $\CISN$ among industrial agents in $N$ under policy $\wp$ such that all the $\ISN$s among promoted groups are implementable (using the provided method in Theorem \ref{theorem_CISN_impl}) and let $E$ be the set of implemented $\ISN$s. For any  $\CISN$,  the incentive values returned by Algorithm \ref{algor2} ensures budget balancedness while preserving fairness (i.e., EFF, SYM, DUM, and ADD).  
\end{proposition}
\begin{proof}  To have budget balancedness, we have to show that the total collectible   tax value (using the provided method in Theorem \ref{theorem_CISN_impl}) is equal to allocated subsidies. (Considering a disposal account---under control of the regulatory agent---for each firm, it is reasonable to assume that \emph{collectible} $\tau$ is equal to \emph{collected} $\tau$.) If the $\CISN$ is $\wp$-compliant, this is obvious as $\tau$ is equal to  zero (thanks to the implementation of all the promoted $\ISN$s). When the $\CISN$ is $\wp$-violating, we use the Shapley value of each agent that contributes to the sub-game of implemented promoted $\ISN$s. As we employ a Shapley-based method, the monetary incentive is budget-balanced thanks to the EFF property and in addition preserves the other three properties (i.e., SYM, DUM, and ADD).      
\end{proof}

Note that the redistribution phase takes place after the implementation of the $\ISN$s and with respect to the evidence set $E$. Otherwise, there will be cases in which the redistribution process provides incentives for agent groups to defect the set of promoted collaborations. Moreover, we highlight that the use of an MC-net representation of games enables calculating the Shapley value in a scalable manner---in Algorithm \ref{algor2} (see \cite{ieong2005marginal} for complexity results). In specific, the Shapley value of an agent in MC-nets is equal to the summation of its Shapley values in each MC-net rule. Accordingly, as in each rule the value can be computed in linear time (in the pattern of the rule), the Shapley value computation for the whole game will be linear in the size of the input. 

\section{Concluding Remarks} \label{sec:conclusions}
This paper provides a coordinated multiagent system---rooted in cooperative game theory---for implementing $\ISN$s that take place under a socio-economic policy. The use of rule-based MC-Net representation of cooperative games enables combining the game with the set of policies and regulations in a natural way. The paper also provides  algorithms that generate regulatory rules to ensure the implementability of ``good'' symbiotic collaborations in the eye of the policy-maker. This extends previous work that merely focused on operational aspects of industrial symbiotic relations--as we introduce the analytical study of the regulatory aspect of $\ISN$s. Finally, it introduces a method for redistribution of collectible tax values. The presented method ensures the budget-balancedness of the monetary incentive mechanism for coordination of $\ISN$s in the implementation phase.

In practice, such a framework supports decision-makers in the $\ISN$ implementation phase by providing operational semantics as tools for reasoning  about the implementability of a given $\ISN$ in a fair and stable manner. Moreover, it supports policy-makers aiming to foster socio-economically desirable  $\ISN$s---by providing  algorithms that generate the required regulatory rules. Finally, it shows that MC-Net is an expressive representation framework for applying normative coordination mechanisms to  multiagent systems.       

This paper focuses on a unique socio-economic policy and a set of rules to ensure it. In this regard, one question that deserves investigation is the possibility of having multiple policies and thus analytical tools for policy option analysis \cite{sara2017structured} in $\ISN$s. Such a framework assists ranking and investigating the applicability of a set of policies in a particular $\ISN$ scenario. Along this line, we aim to  generate  a regulation toolbox  for $\ISN$ policy-makers---since a single regulation may be incapable of ensuring all the desired collaborations under potentially conflicting policies. In that case, possible conflicts among regulations can be resolved using prioritized rule sets (inspired by methods for dealing with potential extensions in argumentation theory \cite{modgil2013general,kaci2008preference}). Accordingly, we will have distinguishable potential $\ISN$ worlds such that in each a maximal set of promoted $\ISN$s are implementable. Another approach to address this is by building upon  multiagent-based Delphi implementations \cite{garcia2008multi}  and generating sets of  nonconflicting, collectively acceptable rules. Here, we consider firm managers as the panel of experts. This can be integrated with agent-oriented representation  methods \cite{garcia2013collection} to analyze how panel transformations affect the implementability  of socio-economic policies. 

In future work, we also aim to focus on administration of $\ISN$s. Then, information-sharing  issues \cite{carter2002reputation,fraccascia2018role} and compliance of involved agents to their commitments will be main concerns for automated trading and business implementations in multiagent industrial symbiosis systems \cite{ash2004empowering}. For that, we plan to model $\ISN$s as normative multiagent organizations  \cite{DBLP:conf/eumas/YazdanpanahYZ16,boissier2013organisational} and integrate data-driven coordination techniques \cite{DBLP:journals/arc/WardiCC18}. Thence, we can rely on norm-aware organization frameworks that focus on operation of normative  organizations \cite{dastani2016commitments,aldewereld2007operationalisation,DBLP:conf/prima/DastaniSY17} to monitor the organization's behavior. Finally, we aim to illustrate the validity of our formally verified framework using realistic case studies and multiagent-based simulations  \cite{fraccascia2017efficacy}. 

\section*{Acknowledgments}
\emph{SHAREBOX} \cite{sharebox}, the project leading to this work, has received funding from the European Union's \emph{Horizon 2020} research and innovation programme under grant agreement No. 680843.

\bibliographystyle{apalike}
\bibliography{mybibfile}

\end{document}